\documentclass{amsart}

\usepackage{graphicx}
\usepackage{xcolor}

\let\labelindent\relax

\usepackage{amssymb,latexsym,amsfonts,amsmath,amsthm}
\usepackage{cite}
\usepackage{textcomp}
\usepackage{algorithmic, algorithm}
\usepackage{enumitem}
\usepackage{hyperref}
\newcommand{\subscript}[2]{$#1 _ #2$}

\def\baselinestretch{1}

\newtheorem{theorem}{Theorem}[section]
\newtheorem{lemma}[theorem]{Lemma}

\newtheorem{proposition}[theorem]{Proposition}

\theoremstyle{definition}
\newtheorem{definition}[theorem]{Definition}

\theoremstyle{remark}
\newtheorem{remark}[theorem]{Remark}

\newcommand{\R}{{\mathbb{R}}}

\DeclareMathOperator{\conv}{conv}

\DeclareMathOperator{\diag}{diag}

\title{Learning to control from expert demonstrations*}

\thanks{*This work was partially supported by the NSF grant 1705135 and by the CONIX Research Center, one of six centers in JUMP, a Semiconductor Research
Corporation (SRC) program sponsored by DARPA.
}

\author[A. Sultangazin, L. Pannocchi, L. Fraile, and P. Tabuada]{Alimzhan Sultangazin, Luigi Pannocchi, Lucas Fraile, and Paulo Tabuada}
\address{Department of Electrical Engineering\\
University of California at Los Angeles,
Los Angeles, CA 90095}
\email{\{asultangazin, lpannocchi, lfrailev, tabuada\}@ucla.edu}

\begin{document}

\maketitle
\thispagestyle{empty}
\pagestyle{empty}
\begin{abstract}
In this paper, we revisit the problem of learning a stabilizing controller from a finite number of demonstrations by an expert. By first focusing on feedback linearizable systems, we show how to combine expert demonstrations into a stabilizing controller, provided that demonstrations are sufficiently long and there are at least $n+1$ of them, where $n$ is the number of states of the system being controlled. When we have more than $n+1$ demonstrations, we discuss how to optimally choose the best $n+1$ demonstrations to construct the stabilizing controller. We then extend these results to a class of systems that can be embedded into a higher-dimensional system containing a chain of integrators. The feasibility of the proposed algorithm is demonstrated by applying it on a CrazyFlie 2.0 quadrotor.
\end{abstract}

\section{Introduction}

\subsection{Motivation}

The usefulness of learning from demonstrations has been well-argued in the literature (see \cite{ Argall2009, Billard2016, Ravichandar2020}). In the context of control, imagine that we need to design a controller for an autonomous car that prioritizes comfort of its passengers. 
It is not obvious how to capture the idea of comfortable driving in a mathematical expression. It is fairly straightforward, however, to collect demonstrations of comfortable driving from human drivers.
There are many other control tasks where providing examples of the desired behaviour is easier than defining such behaviour formally (e.g., teaching a robot to manipulate objects). The growing research interest in learning from demonstrations (LfD) for robot control \cite{Ravichandar2020} reflects the need for a well-defined controller design methodology for such tasks. In this work, we propose a methodology that uses expert demonstrations to construct a stabilizing controller.

There are many examples in the literature, where various LfD methodologies have been applied to robots \cite{Ravichandar2020}. The most popular application of LfD so far is in robotic manipulators. More specifically, LfD is used to teach manipulators skills to perform tasks in manufacturing \cite{Vogt2017}, health-care \cite{vandenBerg2010, Lauretti2017}, and human-robot interaction \cite{Maeda20, RAVICHANDAR2019a}. In addition, LfD has been applied with significant success to ground vehicles \cite{Codevilla2018, Pan2020}, aerial vehicles \cite{Kaufmann2020, Abbeel2010}, bipedal robots \cite{Farchy2013, Mericli2010}, and quadrupedal robots \cite{Kolter2008, Nakanishi2004}. These examples illustrate that, for these platforms, there exist control tasks 
for which LfD techniques are preferable to traditional control approaches.

\subsection{Related work}
In this section, we describe the previous work in learning from demonstrations to indicate where our approach lies within the existing landscape. This is in no way a comprehensive account of the literature on learning from demonstrations, but rather an overview of approaches related to ours 
(please refer to \cite{Ravichandar2020} or \cite{Kroemer2019} for a description of the literature on LfD).

\emph{Policy-learning LfD methods}, to which this work belongs, assume that there exists a mapping from state (or observations) to control input that dictates the expert's behaviour. This mapping is referred to as the expert's policy. The goal of these methods is to find (or approximate) the expert's policy given expert demonstrations. In many machine-learning-based LfD methods, policy learning is viewed as a supervised-learning problem where states and control inputs are treated as features and labels, respectively. We refer to these methods as \emph{behavioural cloning} methods. 
Pioneered in the 80s by works like \cite{Pomerleau1989}, this class of methods is still popular today. Behavioural cloning methods are typically agnostic to the nature of the expert --- demonstrations can be provided by a human (see \cite{Bojarski2016, Codevilla2018}), an offline optimal controller (see \cite{Levine2014, Chen2018}), or a controller with access to privileged state information (see \cite{Chen2019, Kaufmann2020}). They do, however, require a large number of demonstrations to work well in practice and, if trained solely on data from unmodified expert demonstrations, generate unstable policies that cannot recover from drifts or disturbances \cite{Codevilla2018}. The latter problem can be fixed using online meta-algorithms like DAgger \cite{Ross2011} which ensure that training data includes observations of recoveries from perturbations. Using such algorithms, however, comes at the expense of enlarging the training dataset. Moreover, the works on behavioural cloning typically provide few formal stability guarantees and, instead, illustrate performance with experiments. 

Currently, there is a concerted effort to develop policy-learning LfD methods that improve on existing techniques using tools from control theory. 
In that context, the work that is closest to ours is described in \cite{Boyd2020}, where the authors use convex optimization to construct a linear policy that is both close to expert demonstrations and stabilizes a linear system. They guarantee that the resulting controller is optimal with respect to some quadratic cost by adding an additional set of constraints (originally proposed in \cite{Kalman1964}) to the optimization problem. This work has been extended in \cite{Havens2021} to enforce other properties, such as stability, optimality, and $\mathcal{H}_{\infty}$-robustness.
Our methodology is different from those in \cite{Boyd2020} and \cite{Havens2021} because we do not assume the expert to be a linear time-invariant controller.

\subsection{Contributions}
In this paper, we propose a methodology for constructing a controller for a known nonlinear system from a finite number of expert demonstrations of desired behaviour, provided their number exceeds the number of states and the demonstrations are sufficiently long. Our approach consists of two steps:
\begin{itemize}
    \item use feedback linearization to transform the nonlinear system into a chain of integrators;
    \item use affine combinations of demonstrations in the transformed coordinates to construct a control law stabilizing the original system.
\end{itemize}
The expert demonstrations are assumed to be of finite-length, whereas the resulting controller is expected to control the system indefinitely, making this a non-trivial problem to address. In this paper, we formally prove the learned controller asymptotically stabilizes the system. Furthermore, in case there are more demonstrations than states, we determine which subset of demonstrations needs to be chosen to minimize the error between the trajectory of the learned controller and the trajectory of the expert controller. To demonstrate the feasibility of this methodology, we apply it to the problem of quadrotor control. Unlike \cite{Boyd2020}, our methodology produces a controller that is time-varying and \emph{not} linear in the original coordinates. This reflects our belief that, in many cases, the expert demonstration is produced by a nonlinear controller. We also extend the proposed methodology beyond the class of feedback linearizable systems by using the embedding technique described in \cite{Astolfi2018} and demonstrate its feasibility on the classical example of the ball-and-beam system.


A preliminary version of this methodology was introduced in \cite{sultangazin2020}. In \cite{Sultangazin2021}, it was combined together with the data-driven control results from \cite{lucas2020} to learn to control unknown SISO systems from demonstrations. This paper provides a unified presentation of the results from \cite{sultangazin2020}, as well as several new results, such as the discussion on the optimality of the controller approximation error and the extension of the results beyond the class of feedback linearizable systems.

\color{black}
\section{Problem Statement and Preliminaries}
\subsection{Notations and basic definitions}
The notation used in this paper is fairly standard. The integers are denoted by $\mathbb{Z}$, the natural numbers, including zero, by $\mathbb{N}_0$, the real numbers by $\mathbb{R}$, the positive real numbers by $\R^+$, and the non-negative real numbers by $\mathbb{R}^+_0$. We denote by $\|\cdot\|$ (or by $\|\cdot\|_2$) the standard Euclidean norm or the induced matrix 2-norm; and by $\|\cdot \|_F$ the matrix Frobenius norm. A set of vectors $\{v_1, \hdots, v_k\}$ in $\R^n$ is \emph{affinely independent} if the set $\{v_2 - v_1, \hdots, v_k - v_1 \}$ is linearly independent.

A function $\alpha: \R^+_0 \rightarrow \R^+_0$ is of class $\mathcal{K}$ if $\alpha$ is continuous, strictly increasing, and $\alpha(0) = 0$. If $\alpha$ is also unbounded, it is of class $\mathcal{K}_\infty$.
A function $\beta: \R^+_0 \times \R^+_0 \rightarrow \R^+_0$ is of class $\mathcal{KL}$ if, for fixed $t \geq 0$, $\beta(\cdot, t)$ is of class $\mathcal{K}$ and $\beta(r, \cdot)$ decreases to $0$ as $t \rightarrow \infty$ for each fixed $r \geq 0$.


The Lie derivative of a function $h: \R^n \rightarrow \R$ along a vector field $f: \R^n \rightarrow \R^n$, given by $\frac{\partial h}{\partial x} f$, is denoted by $L_fh$. We use the notation $L_f^kh$ for the iterated Lie derivative, i.e., \mbox{$L_f^k h = L_f( L_f^{k-1}h)$}, with $L^0_f h = h$. Given open sets $U \subseteq \R^n$ and $V \subseteq \R^n$, a smooth map $\Phi: U \rightarrow V$ is called a diffeomorphism from $U$ to $V$ if it is a bijection and its inverse $\Phi^{-1}: V \rightarrow U$ is smooth.

Consider the continuous-time system:
\begin{align}
    \label{eq:general_system}
    \dot{x} = f(t,x),
\end{align}
where $x \in \R^n$ is the state and $f: \R^+_0 \times \R^n \rightarrow \R^n$ is a smooth function. The origin of \eqref{eq:general_system} is \emph{uniformly asymptotically stable} if there exist $\beta \in \mathcal{KL}$ and $c > 0$ such that, for all $\|x(t_0)\| < c$, the following is satisfied \cite{Khalil2002}:
\begin{align}
\label{eq:asym_stability}
    \|x(t)\| \leq \beta(\|x(t_0)\|,t-t_0), \quad \forall t \geq t_0 \geq 0.
\end{align}

Consider the continuous-time control system:
\begin{align}
\label{eq:general_system_inputs}
    \dot{x} = f(t,x,u),
\end{align}
where $x \in \R^n$ is the state, $u \in \R^m$ is the input, and \mbox{$f: \R^+_0 \times \R^n \times \R^m \rightarrow \R^n$} is a smooth function. The system \eqref{eq:general_system_inputs} is said to be \emph{input-to-state stable (ISS)} if there exist $\beta \in \mathcal{KL}$ and $\gamma \in \mathcal{K}$ such that for any $x(t_0) \in \R^n$ and any bounded input $u: [t_0, \infty) \rightarrow \R^m$, the following is satisfied:
\begin{align}
\label{eq:input_to_state}
    \| x(t) \| \leq \beta(\|x(t_0)\|, t - t_0) + \gamma\left( \sup_{t_0 \leq \tau \leq t} \| u(\tau) \| \right).
\end{align}

Let $\mathcal{X} = \{ x_1, \hdots, x_k \}$ be a set of points in $\R^n$. A point $x = \sum_{i=1}^k \theta_i x_i$ with $\sum_{i=1}^k \theta_i = 1$ is called an \emph{affine combination} of points in $\mathcal{X}$. If, in addition, $\theta_i \geq 0$ for all $i \in \{1, \hdots, k\}$, then $x$ is a \emph{convex combination} of points in $\mathcal{X}$.


\subsection{Problem Statement}
Consider a known continuous-time control-affine system:
\begin{align}
\label{eq:system}
    \Sigma: \quad \dot{x} = f(x) + g(x)u,
\end{align}
where $x \in \R^n$ and $u \in \R^m$ are the state and the input, respectively; and $f:\R^n \rightarrow \R^n$, $g: \R^n \rightarrow \R^{n \times m}$ are smooth functions. Assume that the origin is an equilibrium point of \eqref{eq:system}. We call a pair $(x,u): \R^+_0 \to \R^n \times \R^m$ a solution of the system \eqref{eq:system} if, for all $t \in \R^+_0$, the equation \eqref{eq:system} is satisfied. Furthermore, we refer to the functions $x$ and $u$ as a trajectory and a control input of the system \eqref{eq:system}.

We say that a controller $k: \R^n \rightarrow \R^m$ is \emph{asymptotically stabilizing} for the system \eqref{eq:system} if the origin is uniformly asymptotically stable for the system \eqref{eq:system} with $u = k(x)$. Suppose there exists an unknown asymptotically stabilizing controller $k$, which we call the expert controller. We assume that $k$ is smooth. Our goal is to learn a controller $\widehat{k}: \R_0^+ \times \R^n \rightarrow \R^m$ such that having $u = \widehat{k}(t,x)$ asymptotically stabilizes the origin of the system \eqref{eq:system}. Towards this goal, we use a set of $M$ finite-length expert solutions \mbox{$\mathcal{D} = \{ (x^i, u^i) \}_{i = 1}^M$} of \eqref{eq:system}, where: for each $i$, the trajectory $x^i: [0, T] \rightarrow \R^n$ and the control input \mbox{$u^i: [0, T] \rightarrow \R^m$} are smooth and satisfy $u^i(t) = k(x^i(t))$ for all $t \in \R_0^+$; $T \in \R_0^+$ is the length of a solution; and $M \geq n+1$. We also ascertain that the ``trivial'' expert solution, wherein $x(t) = 0$ and $u(t) = 0$ for all $t \in [0,T]$, is included in $\mathcal{D}$.
\begin{remark}
In practice, we can record the values of continuous solutions provided by the expert only at certain sampling instants. In this work, however, we choose to work in continuous-time to simplify the theoretical analysis. We can do this without sacrificing practical applicability because it is well-known that continuous-time controller designs can be implemented via emulation and still guarantee stability \cite{Nesic2008}.
\end{remark}

We make the assumption that the system \eqref{eq:system} is feedback linearizable on an open set $U \subseteq \R^n$ containing the origin and the expert demonstrations $x^i(t)$ belong to $U$ for all $t \in [0,T]$. To avoid the cumbersome notation that comes with feedback linearization of multiple-input systems, we assume that $m = 1$, that is, the system \eqref{eq:system} only has a single input. Readers familiar with feedback linearization can verify that all the results extend to multiple-input case, mutatis mutandis (refer to \cite[Ch. 4-5]{Isidori1995} for a complete introduction to feedback linearization). In the single-input case, the system \eqref{eq:system} is feedback linearizable on the open set $U \subseteq \R^n$ if there is an output function \mbox{$h: \R^n \rightarrow \R$} that has relative degree $n$, i.e., for all $x \in U$, $L_gL_f^ih(x) = 0$ for $i = 0, \hdots, n-2$ and $L_gL_f^{n-1} h(x) \neq 0$. 
Moreover, the map: 
\begin{align}
\label{eq:feedback_lin_transformation}
    z = \Phi(x) = \begin{bmatrix} h(x) & L_f h(x) & \cdots & L_f^{n-1} h(x) \end{bmatrix}^T,
\end{align}
is a diffeomorphism from $U$ to its image $\Phi(U)$, i.e., the inverse $\Phi^{-1}: \Phi(U) \rightarrow U$ exists and is also smooth. We further assume, without loss of generality, that $h(0) = 0$.

\section{Learning a stabilizing controller from $n+1$ expert demonstrations}
\label{sec:n+1_demons}

Here, we describe the methodology for constructing an asymptotically stabilizing controller when $M = n+1$. We consider the case when $M \geq n+1$ in Section \ref{sec:more_than_n+1_demons}. 

\subsection{Feedback linearization}
Recall that using the feedback linearizability assumption, we can rewrite the  system dynamics \eqref{eq:system} in the coordinates given by \eqref{eq:feedback_lin_transformation}
resulting in:
\begin{equation}
\label{eq:system_before_feedback}
\begin{aligned}
    \dot{z}_1 &= z_2,\\
    &\vdots\\
    \dot{z}_{n-1} & = z_n,\\
    \dot{z}_n &= a(z) + b(z)u,
\end{aligned}
\end{equation}
where $a = \left( L_f^n h \right)\circ \Phi^{-1}$ and $b = \left( L_g L_f^{n-1}h \right) \circ \Phi^{-1}$. The feedback law:
\begin{equation}
    \begin{aligned}
    \label{eq:feedback}
    u & = {b(z)}^{-1}(-a(z) + v),
    \end{aligned}
\end{equation}
further transforms the system \eqref{eq:system} into the system given by:
\begin{align}
\label{eq:system_brunovsky}
    \dot{z} = Az + Bv,
\end{align}
where $(A,B)$ is a Brunovsky pair.
\begin{remark}
The expert controller $\kappa: \R^n \rightarrow \R$ in the $(z,v)$-coordinates is given by $\kappa(z) = a(z) + b(z)k(\Phi^{-1}(z))$. The smoothness of $k$ implies that the function $\kappa$ is also smooth.
\end{remark}


\subsection{Expert demonstrations}

Recall that the set of demonstrations $\mathcal{D}$ consists of solutions of the system \eqref{eq:system}. Using \eqref{eq:feedback_lin_transformation} and \eqref{eq:feedback}, we can represent the demonstrations $\mathcal{D}$ in  $(z, v)$-coordinates. We denote the resulting set by $\mathcal{D}_{(z,v)} = \{ (z^i, v^i) \}_{i = 1}^M$, where functions $z^i:[0, T] \rightarrow \R^n$ and \mbox{$v^i: [0, T] \rightarrow \R$} are given by:
\begin{align}
\label{eq:transform_z}
    z^i(t) &\triangleq \Phi(x^i(t))\\
\label{eq:transform_v}
    v^i(t) & \triangleq L_f^{n}h(x^i(t)) + L_gL_f^{n-1}h(x^i(t)) u^i(t),
 \end{align}
for all $i \in \{1, \cdots, M \}$ and for all $t \in [0, T]$. We define the set of demonstrations $\mathcal{D}_{(z,v)}$ evaluated at time $t$ as:
\begin{align*}
    \mathcal{D}_{(z,v)}(t) = \{ (z^i(t), v^i(t)) \}_{i=1}^M.
\end{align*}
It can be easily verified that the demonstrations in $\mathcal{D}_{(z,v)}$ satisfy the dynamics \eqref{eq:system_brunovsky} and $v^i(t) = \kappa(z^i(t))$.


\subsection{Constructing the learned controller}
We denote by $\widehat{\kappa}(t,z)$ the controller learned from the expert demonstrations. We begin by partitioning time into intervals of length $T$ and indexing these intervals with $p \in \mathbb{N}_0$. Let us construct the following matrices for $t \in [0, T]$:
\begin{align}
    \label{eq:Z_matrix_easy}
    Z(t) &\triangleq \left[ \begin{array}{@{}c|c|c@{}}  z^{2}(t) - z^{1}(t) & \cdots & z^{{n+1}}(t) - z^{1}(t) \end{array} \right]\\
    \label{eq:V_matrix_easy}
    V(t) &\triangleq \left[ \begin{array}{@{}c|c|c@{}}  v^{2}(t) - v^{1}(t) & \cdots & v^{{n+1}}(t) - v^{1}(t) \end{array} \right].
\end{align}
Our first attempt at constructing the learned controller, which we improve upon later, is to use the piecewise-continuous controller \mbox{$v(t) = \widehat{\kappa}(t,z(pT))$} for all $t \in [pT, (p+1)T]$, where:
\begin{align}
\label{eq:proposed_control_easy}
    \widehat{\kappa}(t,z(pT)) &= V(t-pT)\zeta(p), 
\end{align}
with $\zeta(p) = Z^{-1}(0) z(pT)$, and $Z(t), V(t)$ defined in \eqref{eq:Z_matrix_easy} and \eqref{eq:V_matrix_easy}, respectively.

The next lemma formally shows that an affine combination of trajectories of \eqref{eq:system_brunovsky} is a valid trajectory for \eqref{eq:system_brunovsky}. 

\begin{lemma}
\label{th:solution_linear_combination_easy}
Suppose we are given a set of finite-length solutions \mbox{$\{ (z^i, v^i) \}_{i = 1}^{n+1}$} of the system \eqref{eq:system_brunovsky}, where each $(z^i,v^i)$ is defined for $0 \leq t \leq T$, $T \in \R^+$.
Assume that $\{ z^i(0)\}_{i = 1}^{n+1}$ is an affinely independent set. Then, under the control law $v(t) = V(t-t_0) \zeta$ with $\zeta = Z^{-1}(0) z(t_0)$, the solution of the system \eqref{eq:system_brunovsky} is 
    $z(t) = Z(t-t_0) \zeta$,
for $t_0 \leq t \leq T + t_0$, where $Z(t)$ and $V(t)$ are defined in \eqref{eq:Z_matrix_easy} and \eqref{eq:V_matrix_easy}, respectively. 
\end{lemma}
\begin{proof}
This lemma can be verified by substitution.
\end{proof}
\begin{remark}
Affine independence of the set $\{ z^i(0)\}_{i = 1}^{n+1}$ is a generic property, i.e., this is true for almost all expert demonstrations. In practice, if this set is not affinely independent, a user can eliminate the affinely dependent demonstrations and request the expert to provide additional demonstrations.
\end{remark}

We note, however, that the control law \eqref{eq:proposed_control_easy} samples the state $z$ with a sampling time $T$ and essentially operates in open loop in between these samples. To allow for closed-loop control, we propose the improved controller that has, for all $t \in [pT, (p+1)T]$, the following form:
\begin{equation}
\begin{aligned}
\label{eq:proposed_control_mod_easy}
    v(t) = \widehat{\kappa}(t,z(t)) &= V(t - pT)\zeta(p, t),\\
    \zeta(p, t) &= Z^{-1}(t-pT) z(t).
\end{aligned}
\end{equation}

In the absence of uncertainties and disturbances, by Lemma \ref{th:solution_linear_combination_easy}, the coefficients $\zeta$ satisfy:
\begin{equation}
\begin{aligned}
\label{eq:coeffs_equal_easy}
    \zeta(p,t) &= Z^{-1}(t-pT) z(t) =  Z^{-1}(0) z(pT),
\end{aligned}
\end{equation}
i.e., the controller \eqref{eq:proposed_control_mod_easy} applies the input equal to that applied by the controller \eqref{eq:proposed_control_easy}.

\subsection{Stability of the learned controller}
Assuming \eqref{eq:coeffs_equal_easy} holds, the system \eqref{eq:system_brunovsky} in closed loop with \eqref{eq:proposed_control_mod_easy} has the following form:
\begin{align}
\label{eq:closed_loop_sys_easy}
    \dot{z} = Az + B V(t - pT) Z^{-1}(0) z(pT),
\end{align}
for all $t \in [pT, (p+1)T]$. Integrating the dynamics, we show that the sequence $\{ z(pT)\}_{p \in \mathbb{N}_0}$ satisfies:
\begin{align}
\label{eq:discrete_time_sys_easy}
    z((p+1)T) &= \Psi(T) z(pT),
\end{align}
where:
\begin{align}
\label{eq:monodromy_matrix_easy}
    \Psi(T) \triangleq e^{AT} + \int_0^T e^{A(T-\tau)} B V(\tau) Z^{-1}(0) \mathrm{d}\tau.
\end{align}
By adopting a term from Floquet's theory, we refer to $\Psi(T)$ in \eqref{eq:monodromy_matrix_easy} as the closed-loop monodromy matrix \cite{Kabamba1987}. 

This section's main result provides sufficient conditions for asymptotic stability of system \eqref{eq:system} in closed loop with \eqref{eq:feedback}-\eqref{eq:proposed_control_mod_easy}. 
\begin{theorem}
\label{th:main_result_easy}
Consider the system \eqref{eq:system} and assume it is feedback linearizable on an open set $U \subseteq \R^n$ containing the origin.
Let $T \in \R^+$ and suppose we are given a finite set of demonstrations $\mathcal{D} = \{ (x^i, u^i)\}_{i=1}^{n+1}$ generated by the system \eqref{eq:system}, in closed loop with a smooth asymptotically stabilizing controller $k:\R^n \rightarrow \R$, and satisfying $x^i(t) \in U$ for all $t \in [0,T]$. Assume that $\{ \Phi(x^i(t)) \}_{i=1}^{n+1}$ is affinely independent for all $t \in [0,T]$. Then, there is a $\tilde{T}\in \R^+$ such that for all $T \geq \tilde{T}$, the origin of system \eqref{eq:system} in closed-loop with controller \eqref{eq:feedback}-\eqref{eq:proposed_control_mod_easy} is uniformly asymptotically stable.
\end{theorem}
\begin{proof}
The asymptotic stability of \eqref{eq:system} and \eqref{eq:system_brunovsky} are equivalent on $U$ and $\Phi(U)$ \cite{Sontag1998}, and, therefore, the set $\mathcal{D}_{(z,v)}$ given by \eqref{eq:transform_z} and \eqref{eq:transform_v} also consists of asymptotically stable solutions, i.e., there exists $\beta \in \mathcal{KL}$ such that for all $i \in \{1, ..., n+1\}$:
\begin{align}
\label{eq:demonstrations_asymptotic_stability_easy}
    \|z^i(t)\| \leq \beta(\|z^i(0) \|, t), \quad \forall t \in \R^+_0.
\end{align}

Consider the closed-loop system \eqref{eq:closed_loop_sys_easy}. By Lemma \ref{th:solution_linear_combination_easy}:
\begin{align*}
    z((p+1)T) = Z(T)Z^{-1}(0)z(pT), \quad \forall T \in \R^+_0.
\end{align*}
Combining this with \eqref{eq:discrete_time_sys_easy} implies that:
\begin{align}
\label{eq:monodromy_expr_easy}
    \Psi(T) = Z(T) Z^{-1}(0).
\end{align}

We claim that, for any constants $a,b,c > 0$, there exists $t \in \R^+$ such that $\beta(r,t) < c$ for all $r \in [a, b]$. This claim will be shown using an argument similar to that of the proof of Lemma 16 in \cite{Geiselhart2014}. Using Lemma 4.3 from \cite{Teel2000}, there exist class $\mathcal{K}_\infty$ functions $\sigma_1, \sigma_2$ such that $\beta(r,t) \leq \sigma_1(\sigma_2(r)e^{-t})$ for all $r, t \in \R^+_0$. 
Let $0<\varepsilon<c$. Define, for all $r \in \R^+$, $t(r)$ to be the solution of $\sigma_1(\sigma_2(r) e^{-t}) = c - \varepsilon$ and obtain:
\begin{align*}
    t(r) = -\log \frac{\sigma_1^{-1}(c-\varepsilon)}{\sigma_2(r)}.
\end{align*}
Since $t(r)$ is a continuous function and $[a,b]$ is compact, the extreme value theorem implies that $t^* = \max_{r \in [a,b]} t(r)$ is well-defined. For all $r \in [a,b]$, it is true that:
\begin{align*}
    \beta(r,t^*) &\leq \sigma_1(\sigma_2(r) e^{-t^*}) \leq c - \varepsilon < c.
\end{align*}

Using the previous claim with \mbox{$a = \min_{i \in \{1, ..., n+1\} } \| z^i(0) \|$}, $b = \max_{i \in \{1, ..., n+1 \} } \| z^i(0) \|$ and $c = {1}/\left({2\sqrt{n} \left\lVert Z^{-1}(0) \right\rVert}\right)$, we conclude the existence of $\tilde{T} \in \R^+$ such that, for all $T \geq \tilde{T}$, the following inequality holds:
\begin{align*}
    \beta(\|z^i(0)\|, T) < \frac{1}{2\sqrt{n} \| Z^{-1}(0) \|},
\end{align*}
for all $i \in \{1, \hdots, n+1\}$. Therefore, by \eqref{eq:demonstrations_asymptotic_stability_easy}, for all $i \in \{1, \hdots, n+1\}$ and for all $T \geq \tilde{T}$, we have:
\begin{align}
\label{eq:bound_on_z_easy}
    \|z^i(T)\| < \frac{1}{2\sqrt{n} \| Z^{-1}(0) \|}.
\end{align}
Using \eqref{eq:monodromy_expr_easy} and \eqref{eq:bound_on_z_easy}, for all $T \geq \tilde{T}$, we have:
\begin{equation}
\begin{aligned}
\label{eq:monodromy_stable_easy}
    \|\Psi(T) \| & \leq  \|Z(T)\| \left\lVert Z^{-1}(0) \right\rVert
                 \leq \|Z(T)\|_F \left\lVert Z^{-1}(0) \right\rVert \\
                & = \left( \sum_{i=2}^{n+1} \|z^i(T) - z^{1}(T)\|^2 \right)^{\frac{1}{2}} \left\lVert Z^{-1}(0) \right\rVert \\
                & < \frac{\sqrt{n}}{\sqrt{n} \left\lVert Z^{-1}(0) \right\rVert} \cdot \left\lVert Z^{-1}(0) \right\rVert < 1.
\end{aligned}
\end{equation}

According to stability conditions for linear discrete-time systems (see Theorem 10.9 in \cite{Antsaklis2006}), the equation \eqref{eq:monodromy_stable_easy} implies that, for all $T > \tilde{T}$, the system \eqref{eq:discrete_time_sys_easy} is uniformly exponentially stable. From \cite{Kabamba1987}, we know that uniform exponential stability of the sampled-data system \eqref{eq:discrete_time_sys_easy} implies uniform exponential stability of the system \eqref{eq:system_brunovsky}-\eqref{eq:proposed_control_mod_easy} because the matrices $\Psi(t)$ are bounded for $t \in [0,T]$. Uniform asymptotic stability of the origin for the system \eqref{eq:system_brunovsky}-\eqref{eq:proposed_control_mod_easy} in the $(z,v)$-coordinates implies uniform asymptotic stability of the origin for the feedback equivalent system \eqref{eq:system}-\eqref{eq:feedback}-\eqref{eq:proposed_control_mod_easy} in $(x,u)$-coordinates \cite{Sontag1998}.
\end{proof}
\begin{remark}
Theorem \ref{th:main_result_easy} shows the existence of \mbox{$\tilde{T} \in \R^+$} such that $\|\Psi(T)\|<1$ for all $T \geq \tilde{T}$. In practice, a user can determine $T \in \R^+$ satisfying this condition by directly computing \mbox{$\|\Psi(t)\| = \|Z(t) Z^{-1}(0)\|$} for various $t \in \R^+_0$.
\end{remark}
\begin{remark}
The fact that we assume feedback linearizability on some open set $U \subseteq \R^n$ presents the user with the opportunity to use either local or global feedback linearization results, depending on what their application allows for. We recommend \cite{Isidori1995} as a good starting point to find conditions for both local (see Theorem 4.2.3 in \cite{Isidori1995}) and global (see Theorem 9.1.1 in \cite{Isidori1995}) feedback linearizability.
\end{remark}
\begin{remark}
In Theorem \ref{th:main_result_easy}, we provide a guarantee the learned controller $\widehat{k}$ stabilizes the system at the origin. This result can also be useful when the objective of the learned controller is to track a trajectory. The key idea is to recast the problem of trajectory tracking into that of stabilizing the error dynamics (see Section 4.5 in \cite{Isidori1995}). We consider this generality of the learned controller to be a strength of this approach.
We will experimentally illustrate this in Section \ref{sec:simulation}.
\end{remark}
\begin{remark}
Although we assume in this work an exact knowledge of the state, in most applications, the state is estimated via an observer. Depending on the design of the observer, the stability results of our methodology may also vary. To give an example, using Lemma III.8 from \cite{Sultangazin2021}, we can show that, with a well-designed sampled-data observer providing state estimates of both the expert demonstrations and the current state, we can still retain asymptotic stability. In general, however, a persistent error between the state estimate and the current state can weaken the guarantee of asymptotic stability guarantee of the closed-loop system to that of practical stability. 
\end{remark}

\section{Learning from more than $n+1$ expert demonstrations}
\label{sec:more_than_n+1_demons}
Here, we extend the previous results to the case where more than $M > n+1$. For every interval of length $T$, we show how to select a subset of $n+1$ demonstrations that results in the best approximation of the expert controller. 

\subsection{Preliminaries}
We begin by reviewing several key concepts from multivariate linear interpolation. Let $\mathcal{X} = \{x_1, \hdots, x_k\}$ be a finite set of points in $\R^n$. The \emph{convex hull} of a set $\mathcal{X}$, denoted $\conv \mathcal{X}$, is the set of all convex combinations of points in $\mathcal{X}$ \cite{Boyd2004}. For any $\mathcal{I} \subset \{1, \hdots, k\}$, we define the subset $\mathcal{X}_{\mathcal{I}} = \{x_i \in \mathcal{X} \mid i \in \mathcal{I} \}$. A Cartesian product of two sets $\mathcal{X} \times \mathcal{Y}$ has a natural left projection map $\pi_1: \mathcal{X} \times \mathcal{Y} \rightarrow \mathcal{X}$ (resp., right projection map $\pi_2: \mathcal{X} \times \mathcal{Y} \rightarrow \mathcal{Y}$) given by $\pi_1(x,y) = x$ (resp., $\pi_2(x,y) = y$). An \emph{$n$-simplex} $S$ is the convex hull of a set $\mathcal{X}' = \{x_1', \hdots, x'_{n+1} \}$ of $n+1$ affinely independent points. A \emph{triangulation} of points in $\mathcal{X}$, denoted $\mathcal{T}(\mathcal{X})$, is a collection of $n$-simplices such that their vertices are points in $\mathcal{X}$, their interiors are disjoint, and their union is $\conv \mathcal{X}$. We denote the $n$-simplex in $\mathcal{T}(\mathcal{X})$ containing $x \in \conv \mathcal{X}$ by $S_\mathcal{T}(x)$ and define a vertex index set associated with $x$ in $\mathcal{T}(\mathcal{X})$, denoted $\mathcal{I}_{\mathcal{T}}(x)$, as to satisfy $S_{\mathcal{T}}(x) = \conv \mathcal{X}_{\mathcal{I}_{\mathcal{T}}(x)}$. 
The \emph{Delaunay triangulation} of $\mathcal{X}$, denoted $\mathcal{DT}(\mathcal{X})$, is a triangulation with the property that the circum-hypersphere of every $n$-simplex in the triangulation contains no point from $\mathcal{X}$ in its interior. It is unique if no $n + 1$ points are on the same hyperplane and no $n + 2$ points are on the same hypersphere \cite{Chang2018}.

Let $\psi: \R^n \rightarrow \R^m$ be an unknown function. Given a finite set of points $\mathcal{X} = \{x_1, \hdots, x_k \} \subset \R^n$ and a set of function values $\mathcal{Y} = \{y_1, \hdots, y_k \} \triangleq \{ \psi(x_1), \hdots, \psi(x_k)\}$, an \emph{interpolant} $\widehat{\psi}^{\mathcal{X}, \mathcal{Y}}: \conv \mathcal{X} \rightarrow \R^m$ is an approximation of $\psi$ that satisfies $\widehat{\psi}(x) = \psi(x)$ for all $x \in \mathcal{X}$. We define an interpolant $\widehat{\psi}_{\mathcal{T}}^{\mathcal{X}, \mathcal{Y}}: \conv \mathcal{X} \rightarrow \R^m$, called a \emph{piecewise-linear interpolant} based on $\mathcal{T}(\mathcal{X})$, as:
\begin{align*}
    \widehat{\psi}_{\mathcal{T}}^{\mathcal{X}, \mathcal{Y}}(x) = \sum_{i \in \mathcal{I}_{\mathcal{T}}(x)} \theta_i y_i,
\end{align*}
where $\theta_i \geq 0$ satisfy:
\begin{align*}
    x = \sum_{i \in \mathcal{I}_{\mathcal{T}}(x)} \theta_i x_i, \quad \sum_{i \in \mathcal{I}_{\mathcal{T}}(x)} \theta_i = 1.
\end{align*}

\subsection{Constructing the learned controller}

Let us describe the construction of the controller \mbox{$v=\widehat{\kappa}(t,z)$} for $M \geq n+1$. Define $\mathcal{Z}(t) = \pi_1\left(\mathcal{D}_{(z,v)}(t)\right)$ and \mbox{$\mathcal{V}(t) = \pi_2\left(\mathcal{D}_{(z,v)}(t)\right)$}. We partition time into intervals of length $T$, indexed by $p \in \mathbb{N}_0$. For each \mbox{$[pT, (p+1)T]$}, we propose using the piecewise-continuous control law \mbox{$v(t) = \widehat{\kappa}(t,z(t))$}, where $\widehat{\kappa}(\tau,\xi)$ is defined as follows:
\begin{enumerate}[label=(\roman*)]
    \item For $\xi \in \conv \mathcal{Z}(\tau-pT)$, the value of $\widehat{\kappa}(\tau,\xi)$ is given by the value at $\xi$ of a piecewise-linear interpolant $\widehat{\psi}_{\mathcal{T}}^{\mathcal{Z}(\tau-pT), \mathcal{V}(\tau-pT)}$. Since a piecewise-linear interpolant is determined by an associated triangulation $\mathcal{T}(\mathcal{Z}(\tau - pT))$ \cite{Chang2018}, this implies that there is a family of possible learned controllers we can construct from $\mathcal{D}_{(z,v)}$. Moreover, the value of the interpolant depends only on the values of $\mathcal{Z}_{\mathcal{I}_{\mathcal{T}}(\xi)}(\tau - pT)$ and $\mathcal{V}_{\mathcal{I}_{\mathcal{T}}(\xi)}(\tau - pT)$, where $\mathcal{I}_{\mathcal{T}}(\xi)$ is a vertex set associated with $\xi$ in $\mathcal{T}(\mathcal{Z}(\tau-pT))$.
    \item For $\xi \notin \conv \mathcal{Z}(\tau - pT)$, let $\xi^*$ be the Euclidean projection of $\xi$ onto $\conv \mathcal{Z}(\tau - pT)$. Define the index set \mbox{$\mathcal{I}_{\mathcal{T}}(\xi) = \mathcal{I}_{\mathcal{T}}(\xi^*)$} and express $\xi$ as an affine combination \mbox{$\xi = \sum_{i \in \mathcal{I}_{\mathcal{T}}(\xi)} \theta_i z^i (0)$}. Then, the value of $\widehat{\kappa}(\tau,\xi)$ is given by $
    \widehat{\kappa}(\tau,\xi) = \sum_{i \in \mathcal{I}_T(\xi)} \theta_i v^i(\tau-pT)$.
\end{enumerate}

In both cases, the controller can be concisely expressed if, given a vertex index set $\mathcal{I} = \{ i_1, \hdots, i_{n+1} \}$ for $\mathcal{Z}(t)$ and $\mathcal{V}(t)$, we construct the following matrices:
\begin{align}
    \label{eq:Z_matrix}
    Z_{\mathcal{I}}(t) \triangleq \left[ \begin{array}{@{}c|c|c@{}}  z^{i_2}(t) - z^{i_1}(t) & \cdots & z^{i_{n+1}}(t) - z^{i_1}(t) \end{array} \right]\\
    \label{eq:V_matrix}
    V_{\mathcal{I}}(t) \triangleq \left[ \begin{array}{@{}c|c|c@{}}  v^{i_2}(t) - v^{i_1}(t) & \cdots & v^{i_{n+1}}(t) - v^{i_1}(t) \end{array} \right],
\end{align}
for $t \in [0, T]$. Then, using \eqref{eq:Z_matrix} and \eqref{eq:V_matrix}, the proposed control law, for all $t \in [pT, (p+1)T]$, is given by:
\begin{equation}
\begin{aligned}
\label{eq:proposed_control_mod}
    v(t) = \widehat{\kappa}_{\mathcal{T}}(t,z(t)) &= V_{\mathcal{I}_{\mathcal{T}}(z(t))}(t - pT)\zeta(p, t)\\
    \zeta(p, t) &=  Z^{-1}_{\mathcal{I}_{\mathcal{T}}(z(t))}(t-pT)  z(t).
\end{aligned}
\end{equation}

Note that, in the absence of uncertainties and disturbances, by Lemma \ref{th:solution_linear_combination_easy}, the coefficients satisfy:
\begin{equation}
\begin{aligned}
\label{eq:coeffs_equal}
    \zeta(p,t) &=  Z^{-1}_{\mathcal{I}_{\mathcal{T}}(z(t))}(t-pT)  z(t)\\ &= Z^{-1}_{\mathcal{I}_{\mathcal{T}}(z(pT))}(0) z(pT).
\end{aligned}
\end{equation}
Therefore, for all \mbox{$t \in [pT, (p+1)T]$}, the controller \eqref{eq:proposed_control_mod} applies the input equal to that applied by the following controller:
\begin{equation}
\begin{aligned}
\label{eq:proposed_control}
    v(t) = \widehat{\kappa}_{\mathcal{T}}(t,z(pT)) &= V_{\mathcal{I}_{\mathcal{T}}(z(pT))}(t - pT)\zeta(p) \\
    \zeta(p) &= Z^{-1}_{\mathcal{I}_{\mathcal{T}}(z(pT))}(0)  z(pT).
\end{aligned}
\end{equation}
Incidentally, this corresponds to the value of the piecewise-linear interpolant $\widehat{\psi}^{\mathcal{T}}_{\mathcal{Z}(0), \mathcal{V}(t-pT)}$ at $z(pT)$. 

\subsection{Stability of the learned controller}

Let us define the collection of index sets $\mathcal{P} = \{ \mathcal{I}_1, \hdots, \mathcal{I}_P\}$, where each $\mathcal{I}_j$ selects vertices of an $n$-simplex in $\mathcal{T}(\mathcal{Z}(0))$ and $P = |\mathcal{T}(\mathcal{Z}(0))|$. Note that $\mathcal{P}$ is a finite set because there are only finitely many $n$-simplices in $\mathcal{T}(\mathcal{Z}(0))$.
Suppose the index set associated with $z(pT)$ in $\mathcal{T}(\mathcal{Z}(0))$ is $\mathcal{I}_{\mathcal{T}}(z(pT)) = \mathcal{I}_{j(p)}$ for some $j(p) \in \{1, \hdots, P \}$. Assuming \eqref{eq:coeffs_equal} holds, the system \eqref{eq:system_brunovsky} in closed loop with \eqref{eq:proposed_control_mod} is given by:
\begin{align}
\label{eq:closed_loop_sys}
    \dot{z} = Az + B V_{\mathcal{I}_{j(p)}}(t - pT) Z^{-1}_{\mathcal{I}_{j(p)}}(0) z(pT),
\end{align}
for all $t \in [pT, (p+1)T]$. Integrating the dynamics shows that the sequence $\{ z(pT)\}_{p \in \mathbb{N}_0}$ satisfies:
\begin{align}
\label{eq:discrete_time_sys}
    z((p+1)T) &= \Psi_{j(p)}(T) z(pT),
\end{align}
where 
\begin{align*}
    \Psi_{j(p)}(T) \triangleq e^{AT} + \int_0^T e^{A(T-\tau)} B V_{\mathcal{I}_{j(p)}}(\tau) Z^{-1}_{\mathcal{I}_{j(p)}}(0) \mathrm{d}\tau.
\end{align*}
Note that now, instead of a single monodromy matrix, we have a set of monodromy matrices $\{\Psi_j(T)\}_{j=1}^P$.

The following result is an extension of Theorem \ref{th:main_result_easy} for \mbox{$M \geq n+1$} demonstrations.
\begin{theorem}
\label{th:main_result}
Consider the system \eqref{eq:system} and assume it is feedback linearizable on an open set $U \subseteq \R^n$ containing the origin. 
Let $T \in \R^+$ and suppose we are given a finite set of demonstrations $\mathcal{D} = \{ (x^i, u^i)\}_{i=1}^{n+1}$ generated by the system \eqref{eq:system}, in closed loop with a smooth asymptotically stabilizing controller $k:\R^n \rightarrow \R$, and satisfying $x^i(t) \in U$ for all $t \in [0,T]$. Assume that $\{ \Phi(x^i(t)) \}_{i=1}^{M}$ is affinely independent for all $t \in [0,T]$. Then, there exists a $\tilde{T}\in \R^+$ such that for all $T \geq \tilde{T}$, the origin of system \eqref{eq:system} in closed-loop with controller \eqref{eq:feedback}-\eqref{eq:proposed_control_mod} is uniformly asymptotically stable.
\end{theorem}
\begin{proof}
The proof of Theorem \ref{th:main_result_easy} implies the existence of $\tilde{T}_j \in \R$ such that $\|\Psi_j(t) \| < 1$
for all $t \geq \tilde{T}_j$. We choose $\tilde{T} = \max_{j \in \{1, \hdots, P \}} T_j$. The system \eqref{eq:system_brunovsky} in closed loop with controller \eqref{eq:proposed_control_mod} can be represented as a switched system \eqref{eq:discrete_time_sys}, where $j(p) \in \{ 1, \hdots, P \}$ is a switching sequence. By Theorem 3 in \cite{Zhai2002}, the fact that $\|\Psi_j(T)\|<1$ for all $T \geq \tilde{T}$ and $j \in \{1, \hdots, P\}$ implies that, for any switching signal $j(p)$, the system \eqref{eq:discrete_time_sys} is uniformly exponentially stable. Since the matrices $\Psi_j(t)$ are bounded for $t \in [0,T]$, the system \eqref{eq:system_brunovsky} in closed loop with controller \eqref{eq:proposed_control_mod} is uniformly exponentially stable. Uniform asymptotic stability of the origin for the system \eqref{eq:system_brunovsky}-\eqref{eq:proposed_control_mod} in the $(z,v)$-coordinates implies uniform asymptotic stability of the origin for the feedback equivalent system \eqref{eq:system}-\eqref{eq:feedback}-\eqref{eq:proposed_control_mod} in $(x,u)$-coordinates \cite{Sontag1998}.
\end{proof}

\subsection{Optimality of the learned controller}

Recall that the piecewise-linear interpolant defining the controller $\widehat{\kappa}_{\mathcal{T}}$ depends on the choice of the triangulation $\mathcal{T}(\mathcal{Z}(t-pT))$. Assuming \eqref{eq:coeffs_equal} holds, this choice reduces to the choice of the triangulation $\mathcal{T}(\mathcal{Z}(0))$, which dictates the index set of demonstrations $\mathcal{I}_{\mathcal{T}}(z(pT))$ used to construct the solution for each interval $[pT, (p+1)T]$. Without loss of generality, in what follows we discuss the solutions on the interval $[0,T]$ only --- a solution on $[pT, (p+1)T]$ can be represented as a solution on $[0,T]$ with the initial condition equal to $z(pT)$. 

Typically, there are several triangulations one can define given a set of sample points $\mathcal{Z}(0)$. We want our choice of triangulation to result in closed-loop trajectories that approximate expert trajectories well for any initial state $z_0 \in \conv \mathcal{Z}(0)$ distinct from $\mathcal{Z}(0)$. More precisely, we want to find a triangulation $\mathcal{T}(\mathcal{Z}(0))$ that best approximates the function $\phi: [0,T] \times \conv \mathcal{Z}(0) \rightarrow \R^n$, which defines solutions of \eqref{eq:system_brunovsky} under the expert controller $\kappa$, by the function $\widehat{\phi}_{\mathcal{T}}: [0,T] \times \conv \mathcal{Z}(0) \rightarrow \R^n$, which defines the solutions of \eqref{eq:system_brunovsky} under the learned controller $\widehat{\kappa}_{\mathcal{T}}$. That is, we want solution to:
\begin{align}
\label{eq:best_worst_problem}
    \min_{\mathcal{T}(\mathcal{Z}(0))} \sup_{\phi \in \mathcal{F}} \max_{t \in [0,T]} \left\Vert \phi(t,z_0) - \widehat{\phi}_{\mathcal{T}}(t,z_0)  \right\Vert,
\end{align}
where $\mathcal{F}$ is the class of functions to which the expert solutions belong. We can view \eqref{eq:best_worst_problem} as a game where we pick $\mathcal{T}(\mathcal{Z}(0))$, and the adversary, upon seeing our choice of $\mathcal{T}(\mathcal{Z}(0))$, picks $\phi$ to maximize the cost.

Let us leverage the properties $\phi(t,z_0)$ has by virtue of describing solutions of \eqref{eq:system_brunovsky} under the expert controller $\kappa$ to determine the class $\mathcal{F}$. We will use the notation $\phi_t: \R^n \rightarrow \R^n$ for $\phi_t(z_0) = \phi(t,z_0)$. By Theorem 4.1 in \cite[Ch. V]{Hartman2002}, since $\kappa$ is a smooth function, the Hessians of the coordinate functions of the solution $\frac{\partial^2 \phi_i}{\partial z_0^2}(t,z_0)$ are continuous with respect to $t$ and $z_0$. By the extreme value theorem, compactness of $\conv \mathcal{Z}(0)$ implies that, for every $i$, there exists $H \in \R^+_0$ such that $\left\Vert\frac{\partial^2 \phi_i}{\partial z_0^2}(t,z_0)\right\Vert \leq H_i$ for all $t \in [0,T]$ and $z_0 \in \conv \mathcal{Z}(0)$. Thus, the norms of the Hessians of the coordinate functions can be bounded by $H = \max \{ H_1, \hdots, H_n \}$. We denote the class of functions whose coordinate functions have the Hessian norm smaller or equal to $H$ by $\mathcal{F}(H)$. For a fixed $t \in [0,T]$, $\phi_t \in \mathcal{F}(H)$ and, therefore, the function $\phi$ belongs to $\mathcal{F}(H)^{[0,T]}$, the set of all functions from $[0,T]$ to $\mathcal{F}(H)$.
\begin{definition}
\label{def:smallest_worst_case_approx}
For any $z_0 \in \conv \mathcal{Z}(0)$ and any learned controller $\kappa_\mathcal{T}$, the worst-case trajectory approximation error on the interval $[0,T]$ is given by:
\begin{equation*}
\begin{aligned}
\sup_{\phi \in \mathcal{F}(H)^{[0,T]}} \max_{t \in [0,T]} \left\Vert \phi(t,z_0) - \widehat{\phi}_{\mathcal{T}}(t,z_0)  \right\Vert,
\end{aligned}
\end{equation*}
where $\phi: [0,T] \times \R^n \rightarrow \R^n$ is the trajectory of the system \eqref{eq:system_brunovsky} with the initial condition $z_0$ under the expert controller $\kappa$, $\widehat{\phi}_{\mathcal{T}}: [0,T] \times \R^n \rightarrow \R^n$ is the trajectory of the system \eqref{eq:system_brunovsky} with the same initial condition $z_0$ under the learned controller $\widehat{\kappa}_\mathcal{T}$, and $\mathcal{F}(H)^{[0,T]}$ is the set of all functions from $[0,T]$ to $\mathcal{F}(H)$. The smallest worst-case trajectory approximation error on the interval $[0,T]$ is given by:
\begin{equation}
\begin{aligned}
\label{eq:best_worst_trajectory}
\min_{\mathcal{T}(\mathcal{Z}(0))}\sup_{\phi \in \mathcal{F}(H)^{[0,T]}} \max_{t \in [0,T]} \left\Vert \phi(t,z_0) - \widehat{\phi}_{\mathcal{T}}(t,z_0)  \right\Vert.
\end{aligned}
\end{equation}
\end{definition}
The following lemma by Omohundro \cite{Omohundro1990} shows that the Delaunay triangulation leads to the best worst-case piecewise-linear interpolation for functions in $\mathcal{F}(H)$. For an efficient implementation of piecewise-linear interpolation based on the Delaunay triangulation, we refer the reader to \cite{Chang2018}.
\begin{lemma}[\hspace{-0.1mm}\cite{Omohundro1990}]
    \label{th:delauney_optimal}
    Let $\psi: \R^n \rightarrow \R^m$ satisfy the bounded Hessian norm property, i.e., \mbox{$\psi \in \mathcal{F}(H)$}, for some $H \in \R^+_0$.
    Given a set of points $\mathcal{X} = \{x_1, \hdots, x_k\} \subset \R^n$ and a set of function values $\mathcal{Y} = \{y_1, \hdots, y_k\} \subset \R^m$, the piecewise-linear interpolant with the smallest maximum approximation error is based on the Delaunay triangulation $\mathcal{DT}(\mathcal{X})$, i.e., for any point $x \in \conv \mathcal{X}$, the following is true:
    \begin{align*}
         \left\Vert \psi(x) - \widehat{\psi}_{\mathcal{DT}}^{\mathcal{X}, \mathcal{Y}}(x) \right\Vert = 
         \min_{\mathcal{T}(\mathcal{X})} \max_{\psi \in \mathcal{F}(H)}
        \left\Vert \psi(x) - \widehat{\psi}_{\mathcal{T}}^{\mathcal{X}, \mathcal{Y}}(x) \right\Vert.
    \end{align*}
\end{lemma}

The following proposition uses Lemma \ref{th:delauney_optimal} to show that choosing the Delaunay triangulation defines the learned controller that results in closed-loop trajectories that best approximate the corresponding expert trajectories. 

\begin{proposition}
\label{th:best_worst_case}
Consider the system \eqref{eq:system} and assume it is feedback linearizable on an open set $U \subseteq \R^n$ containing the origin. 
Let $T \in \R^+$ and suppose we are given a finite set of demonstrations $\mathcal{D} = \{ (x^i, u^i)\}_{i=1}^{n+1}$ generated by the system \eqref{eq:system}, in closed loop with a smooth asymptotically stabilizing controller $k:\R^n \rightarrow \R$, and satisfying $x^i(t) \in U$ for all $t \in [0,T]$. Assume that $\{ \Phi(x^i(t)) \}_{i=1}^{M}$ is affinely independent for all $t \in [0,T]$. For any $z_0 \in \conv \mathcal{Z}(0)$, the controller $\widehat{\kappa}_{\mathcal{DT}}$ based on the Delaunay triangulation $\mathcal{DT}(\mathcal{Z}(0))$ defined as in \eqref{eq:proposed_control_mod} results in closed-loop trajectories in $z$-coordinates that have the smallest worst-case trajectory approximation error on the interval $[0,T]$ as defined in Definition \ref{def:smallest_worst_case_approx}.
\end{proposition}
\begin{proof}
Recall that by Lemma \ref{th:solution_linear_combination_easy} the trajectory of the system \eqref{eq:system_brunovsky} under the learned controller \eqref{eq:proposed_control_mod} is given by:
\begin{align*}
    \widehat{\phi}_{\mathcal{T}}(t,z_0) = Z_{\mathcal{I}_\mathcal{T}(z_0)} (t) \zeta,
\end{align*}
where $\zeta \in \R^n$ is a vector of affine coefficients that we choose $\zeta$ at the beginning of the interval $[0,T]$ and keep constant.

For a fixed $t \in [0, T]$, we can interpret $\widehat{\phi}_{\mathcal{T}}(t, \cdot)$ as a piecewise-linear interpolant of $\phi_t$ mapping initial conditions to the state reached at time $t$ based on sample points $\mathcal{Z}(0)$ and sample values $\mathcal{Z}(t)$. Therefore, since, for any $t \in [0, T]$, the function $\phi_t \in \mathcal{F}(H)$, by Lemma \ref{th:delauney_optimal}, the function $\widehat{\phi}_{\mathcal{DT}}(t, \cdot)$ is the best worst-case approximation of the function $\phi_t$, i.e.:
\begin{equation}
\begin{aligned}
\label{eq:delauney_is_best}
    \sup_{\phi_t \in \mathcal{F}(H)} &\left\lVert \phi_t(z_0) - \widehat{\phi}_{\mathcal{DT}}(t,z_0) \right\rVert \leq \\
    &\quad \quad \sup_{\phi \in \mathcal{F}(H)} \left\Vert \phi_t(z_0) - \widehat{\phi}_{\mathcal{T}}(t,z_0) \right\Vert,
\end{aligned}
\end{equation}
for any triangulation $\mathcal{T}(Z(0))$ and any $z_0 \in \conv \mathcal{Z}(0)$.
Noting that \eqref{eq:delauney_is_best} holds for all $t \in [0,T]$, we have:
\begin{equation*}
\begin{aligned}
    \max_{t \in [0, T] } &\sup_{\phi \in \mathcal{F}(H)^{[0,T]}} \left\Vert \phi( t,z_0) - \widehat{\phi}_{\mathcal{DT}}( t,z_0)\right\Vert \leq \\
    &\quad \quad \quad \quad \max_{t \in [0, T] } \sup_{\phi \in \mathcal{F}(H)^{[0,T]}} \left\Vert \phi( t ,z_0) - \widehat{\phi}_{\mathcal{T}}( t ,z_0) \right\Vert,
\end{aligned}
\end{equation*}
that can be written as:
\begin{equation*}
\begin{aligned}
    \sup_{\phi \in \mathcal{F}(H)^{[0,T]}} \max_{t \in [0, T] }& \left\Vert \phi( t,z_0) - \widehat{\phi}_{\mathcal{DT}}( t,z_0)\right\Vert \leq \\
    & \quad \sup_{\phi \in \mathcal{F}(H)^{[0,T]}} \max_{t \in [0, T] } \left\Vert \phi( t ,z_0) - \widehat{\phi}_{\mathcal{T}}( t ,z_0) \right\Vert.
\end{aligned}
\end{equation*}
\end{proof}

\begin{remark}
While we justify the construction of the controller for $z(pT) \in \conv \mathcal{Z}(0)$ with optimality in terms of approximation error, we cannot provide a similar justification for $z(pT) \notin \conv \mathcal{Z}(0)$. Therefore, we suggest collecting the expert demonstrations in such a way that the normal region of operation belongs to the convex hull of the demonstrations.
\end{remark}

\begin{remark}
Note that the metric we use to formulate the error in  \eqref{eq:best_worst_trajectory} is expressed in $z$-coordinates instead of the original $x$-coordinates. While we cannot generally have a guarantee that the best worst-case approximation in the $z$-coordinates translates to that in the $x$-coordinates, the metric used in \eqref{eq:best_worst_trajectory} and the Euclidean norm metric in the $x$-coordinates are strongly equivalent on $\conv \mathcal{Z}(t-pT)$ due to the Lipschitz continuity of $\Phi$ and its inverse.
\end{remark}

\begin{remark}
Similarly to Proposition \ref{th:best_worst_case}, one can use Lemma \ref{th:delauney_optimal} to show that the learned controller $\widehat{\kappa}_{\mathcal{DT}}$ based on the Delaunay approximation $\mathcal{DT}(\mathcal{Z}(0))$ is the best worst-case approximation of the expert controller $\kappa$, i.e., for any $z \in \conv \mathcal{Z}(0)$, the Delaunay triangulation is the solution of:
\begin{equation}
\begin{aligned}
\min_{\mathcal{T}(\mathcal{Z}(0))} \max_{\kappa \in \mathcal{F}(H')} \left\Vert \kappa(z) - \widehat{\kappa}_{\mathcal{T}}(z)  \right\Vert,
\end{aligned}
\end{equation}
where $H' \in \R$ is a bound on the Hessian norms of coordinate functions of $\kappa$.
\end{remark}

\section{Learning a stabilizing controller for non-feedback linearizable systems}
\label{sec:non_feedback}

In Sections \ref{sec:n+1_demons} and \ref{sec:more_than_n+1_demons}, we propose a methodology for learning control from expert demonstrations assuming the system is feedback linearizable. Here, we extend our methodology to systems outside of the class of feedback linearizable systems using an embedding technique described in \cite{Astolfi2018}.

\subsection{Embedding technique}

First, we describe the embedding technique from \cite{Astolfi2018}. This technique immerses a nonlinear system of dimension $n$ into an extended system that contains a chain of $n$ integrators via dynamic feedback. Although in \cite{Astolfi2018} only single-input single-output systems were considered, it can be shown that a similar technique applies to multiple-input multiple-output systems. For clarity of exposition, however, we will consider a system \eqref{eq:system} with $m = 1$ and the results extend to multiple-input multiple-output case, mutatis mutandis.

Given constants $w_j \in \R$, $j = 1, \hdots, n-1$ and an output map $h: \R^n \rightarrow \R$, we define $\Phi: \R^{2n-1} \rightarrow \R^{2n-1}$ by:
\begin{equation}
\label{eq:transformation}
    \begin{aligned}
    \begin{bmatrix}
    z \\
    \xi
    \end{bmatrix} =
    \Phi(x,\xi) =
    \begin{bmatrix}
    \Phi_z(x,\xi) \\
    \xi
    \end{bmatrix} =
    \begin{bmatrix}
    h(x) + \xi_1\\
    L_f h(x) + \xi_2 \\
    L^2_f h(x) + \xi_3 \\
    \vdots \\
    L_f^{n-1} h(x) - \sum_{j=1}^{n-1} w_j \xi_j \\
    \xi
    \end{bmatrix},
    \end{aligned}
\end{equation}
where $\xi \in \R^{n-1}$. We also define the auxiliary dynamics:
\begin{equation}
\label{eq:aux_dynamics}
    \begin{aligned}
    \dot{\xi}_1 &= \xi_2 - L_g h(x) u,\\
    \dot{\xi}_2 &= \xi_3 - L_g L_f h(x) u, \\
    &\vdots \\
    \dot{\xi}_{n-1} &= - \sum_{i=1}^{n-1}w_i \xi_i - L_g L_f^{n-2} h(x) u,
    \end{aligned}
\end{equation}
and the feedback law:
\begin{equation}
\label{eq:dynamic_control}
u = \frac{1}{r(x)}\left(s(x, \xi) + v \right)
\end{equation}
where:
\begin{align}
\label{eq:denom_astolfi}
    r(x) = L_gL_f^{n-1}h(x) + \sum_{j=1}^{n-1} w_j L_g L_f^{j-1}h(x),
\end{align} 
and:
\begin{align}
    s(x, \xi) = -L_f^n h(x) + \sum_{i=1}^{n-1} w_i w_{n-1} \xi_i - \sum_{j=1}^{n-2} w_j \xi_{j+1}.
\end{align}
We say that the system \eqref{eq:system} is \emph{feedback linearizable through an embedding} on the open set $U \subseteq \R^n$ if there exist constants $w_j$, $j = 1, \hdots, n-1$ and the output map $h$ such that $\Phi$ is a diffeomorphism from $U$ to $\Phi(U)$ and $r(x) \neq 0$ for all $x \in U$.

If the system \eqref{eq:system} is feedback linearizable through an embedding, we can rewrite the dynamics of \eqref{eq:system} and \eqref{eq:aux_dynamics} in the $(z,\xi)$-coordinates given by \eqref{eq:transformation} resulting in the system that consists of the subsystem describing evolution of $z$ given by:
\begin{equation}
\label{eq:system_after_diff}
\begin{aligned}
    \dot{z}_1 &= z_2,\\
    &\vdots\\
    \dot{z}_{n-1} &= z_n,\\
    \dot{z}_n &= -s(x,\xi) + r(x)u,
\end{aligned}
\end{equation}
and the subsystem describing the evolution of $\xi$ given by:
\begin{equation}
\label{eq:aux_dynamics2}
\begin{split}
    \dot{\xi}_1 &= \left( \xi_2 - L_g h(x)u  \right)_{(x,\xi) = \Phi^{-1}(z,\xi)} \\
    &\vdots \\
    \dot{\xi}_{n-1} &= \left( - \sum_{i=1}^{n-1}w_i \xi_i - L_g L_f^{n-2} h(x)u \right)_{(x,\xi) = \Phi^{-1}(z,\xi)}.
\end{split}
\end{equation}
Furthermore, when the system \eqref{eq:system} is feedback linearizable through an embedding, the feedback law given by \eqref{eq:dynamic_control} is well-defined and transforms the $z$-subsystem \eqref{eq:system_after_diff} into the chain of integrators given by \eqref{eq:system_brunovsky} and the $\xi$-subsystem \eqref{eq:aux_dynamics2} into:

\small
\begin{equation}
\label{eq:xi_subsystem_transformed}
\begin{split}
    \dot{\xi}_1 &= \left( \xi_2 - \frac{L_g h(x)}{r(x)}(s(x,\xi) + v) \right)_{(x,\xi) = \Phi^{-1}(z,\xi)} \\
    &\vdots \\
    \dot{\xi}_{n-1} &= \left( - \sum_{i=1}^{n-1}w_i \xi_i - \frac{L_g L_f^{n-2} h(x)}{r(x)} (s(x,\xi) + v) \right)_{(x,\xi) = \Phi^{-1}(z,\xi)}.
\end{split}
\end{equation}
\normalsize

\subsection{Expert demonstrations}

Similarly to the case of fully feedback linearizable systems, we assume that the system \eqref{eq:system} is feedback linearizable through an embedding on an open set $U \subseteq \R^n$ containing the origin and the demonstrations $x^i(t)$ belong to $U$ for all $t \in [0,T]$. We first transform the demonstrations $\mathcal{D}$ into $(z, \xi, v)$-coordinates. For each demonstration $(x^i, u^i)$, we use \eqref{eq:transformation}, \eqref{eq:aux_dynamics}, and \eqref{eq:dynamic_control} to transform $(x^i, u^i)$ into $(z^i, \xi^i, v^i)$. More specifically:
\begin{itemize}
    \item we choose an arbitrary $\xi_0 \in \R^{n-1}$ to initialize \mbox{$\xi^i(0) = \xi_0$}, and solve the equation in \eqref{eq:aux_dynamics} using demonstrations $(x^i, u^i)$ as the input to determine $\xi^i(t)$ for $t \in [0,T]$;
    \item for all $t \in [0, T]$, using $\xi^i(t)$, we determine $z^i(t)$ from \eqref{eq:transformation} and $v^i(t)$ from \eqref{eq:dynamic_control} as:
\begin{align}
\label{eq:transform_z_astolfi}
    z^i(t) &= \Phi(x^i(t), \xi^i(t)) \\
\label{eq:transform_v_astolfi}
    v^i(t) &= r(x^i(t)) u^i(t) - s(x^i(t), \xi^i(t)).
\end{align}
\end{itemize}
We denote the resulting set of demonstrations by:
\begin{align}
\label{eq:demonstrations_astolfi}
    \mathcal{D}_{(z,\xi,v)} = \{ (z^1, \xi^1, v^1), \hdots, (z^n, \xi^n, v^n) \},
\end{align}
where $z^i: [0, T] \rightarrow \R^n$, $\xi^i: [0, T] \rightarrow \R^{n-1}$, and \mbox{$v^i: [0, T] \rightarrow \R$}.
\color{black}
\subsection{Constructing the learned controller for the extended class of systems}

We now show that, for $M = n + 1$, the controller \mbox{$v = \widehat{\kappa}(t,z)$} from \eqref{eq:proposed_control_mod_easy} stabilizes the system \eqref{eq:system} by stabilizing the chain of integrators \eqref{eq:system_brunovsky} in the transformed coordinates $(z, \xi)$. Please note that we focus on the case $M = n+1$ for ease of exposition, and the proposed extension is also compatible with the case $M \geq n+1$ described in Section \ref{sec:more_than_n+1_demons}.

The following statement provides sufficient conditions for stability of \eqref{eq:system} under the control law \eqref{eq:proposed_control_mod_easy}.

\begin{theorem}
\label{th:main_result_astolfi}
Consider the system \eqref{eq:system} and assume it is feedback linearizable through an embedding on an open set $U \subseteq \R^n$ containing the origin.
Let $T \in \R^+$ and suppose we are given a finite set of demonstrations $\mathcal{D}= \{ (x^i, u^i)\}_{i=1}^{n+1}$ of the system \eqref{eq:system} generated by the system \eqref{eq:system}, in closed loop with a smooth asymptotically stabilizing controller $k: \R^n \rightarrow \R$, and
satisfying $x^i(t) \in U$ for all $t \in [0,T]$. Further, suppose the following two conditions hold:
\begin{enumerate}[label=(\subscript{A}{{\arabic*}})]
\item the matrix:
\begin{align}
    A_\xi = \begin{bmatrix}
    0 & 1 & \cdots & 0\\
    \vdots & \vdots & \ddots & \vdots \\
    0 & 0 & \cdots & 1 \\
    -w_1 & -w_2 & \cdots & -w_{n-1}
    \end{bmatrix}
\end{align}
    is Hurwitz;
\item the $\xi$-subsystem in \eqref{eq:xi_subsystem_transformed} is input-to-state stable (ISS) with respect to $z$ and $v$.
\end{enumerate}
Assume that the set $\{ \Phi_z(x^1(t)), \hdots, \Phi_z(x^n(t))\}$ is affinely independent for all $t \in [0,T]$. Then, there exists a $\tilde{T}\in \R^+$ such that for all $T \geq \tilde{T}$, the origin of system \eqref{eq:system} in closed-loop with the auxiliary dynamics \eqref{eq:aux_dynamics} and controller \eqref{eq:dynamic_control}-\eqref{eq:proposed_control_mod_easy} is uniformly asymptotically stable.
\end{theorem}
\begin{proof}
Let us use condition $(A_1)$ to show that the expert solutions $\{(x^i, u^i)\}_{i=1}^{n+1}$ are uniformly asymptotically stable. Since $k(x)$ is asymptotically stabilizing, the origin of the system \mbox{$\dot{x} = f(x) + g(x)k(x)$} is uniformly asymptotically stable. Because the origin is the equilibrium point of \mbox{$\dot{x} = f(x) + g(x)k(x)$}, we have that $k(0) = 0$. The $\xi$-subsystem in \eqref{eq:aux_dynamics} with $u = k(x)$ can be interpreted as a control system with the input $x$. Condition $(A_1)$ together with the fact that $k(0) = 0$ implies that the $\xi$-subsystem in \eqref{eq:aux_dynamics} with $u = k(x)$ is uniformly exponentially stable when $x \equiv 0$. Uniform exponential stability of the unforced $\xi$-subsystem implies that the $\xi$-subsystem in \eqref{eq:aux_dynamics} with $u = k(x)$ is ISS with respect to $x$ (see Lemma 4.6 in \cite{Khalil2002}). By Lemma 4.7 in \cite{Khalil2002}, input-to-state stability of the $\xi$-subsystem with $u = k(x)$ with respect to $x$ as input and uniform asymptotic stability of $\dot{x} = f(x) + g(x) k(x)$ implies that there is a class $\mathcal{KL}$ function $\beta: \R_0^+ \times \R_0^+ \rightarrow \R_0^+$ such that for all $i \in \{ 1, \hdots, n \}$:
\begin{equation}
\label{eq:bound_x_xi}
\begin{aligned}
    \left\lVert \begin{bmatrix} x^i(t) \\ \xi^i(t) \end{bmatrix} \right\rVert \leq \beta \left( \left\lVert \begin{bmatrix} x^i(0) \\ \xi^i(0) \end{bmatrix} \right\rVert,t \right).
\end{aligned}
\end{equation}
Because the system \eqref{eq:system} is feedback linearizable through an embedding, we have that $\Phi$ given by \eqref{eq:transformation} is a diffeomorphism. Therefore, according to \cite{Sontag1998}, the inequality \eqref{eq:bound_x_xi} implies that for all $i \in \{ 1, \hdots, n \}$:
\begin{align}
\label{eq:bounds_on_demonstrations}
    \left\lVert \begin{bmatrix}
    z^i(t) \\ \xi^i(t) \end{bmatrix} \right\rVert &\leq \beta_1 \left( \left\lVert \begin{bmatrix} z^i(0) \\ \xi^i(0) \end{bmatrix} \right\rVert, t \right),
\end{align}
where $\beta_1:\R_0^+ \times \R_0^+ \rightarrow \R_0^+$ is a class $\mathcal{KL}$ function.

We can use \eqref{eq:bounds_on_demonstrations} to show that there is a $\tilde{T} \in \R^+_0$ such that:
\begin{align*}
    \beta_1 \left( \left\lVert \begin{bmatrix} z^i(0) \\ \xi^i(0) \end{bmatrix} \right\rVert, T \right) < \frac{1}{2 \sqrt{n} \|Z^{-1}(0) \|},
\end{align*}
for all $T \geq \tilde{T}$. Using the argument from the proof of Theorem \ref{th:main_result_easy} allows us to conclude uniform exponential stability of the origin of the $z$-subsystem given by \eqref{eq:system_brunovsky}, provided $T > \tilde{T}$.


The uniform exponential stability of the $z$-subsystem implies that there is a function $\beta_z \in \mathcal{KL}$ such that:
\begin{align}
\label{eq:bound_z_exp}
    \|z(t)\| \leq \beta_z(\|z(0)\|, t).
\end{align}
The matrix product $V(t)Z^{-1}(t)$ is continuous with respect to $t$ and defined on $[0,T]$ that is compact. By the extreme value theorem, this product has a bounded norm. This fact, together with the inequality \eqref{eq:bound_z_exp}, implies that the control input $v  \widehat{\kappa}(t,z)$ given by \eqref{eq:proposed_control_mod_easy} satisfies:
\begin{equation}
\begin{aligned}
\label{eq:bound_on_v}
    \|v(t) \| \leq \beta_v(\|z(0)\|, pT),
\end{aligned}
\end{equation}
where $\beta_v$ is also a class $\mathcal{KL}$ function.

By condition $(A_2)$, the $\xi$-subsystem in \eqref{eq:xi_subsystem_transformed} is ISS with respect to $z$ and $v$. 
Lemma 4.7 in \cite{Khalil2002} shows that the ISS property, along with the bounds \eqref{eq:bound_z_exp} and \eqref{eq:bound_on_v}, allows us to conclude that the origin of the system \eqref{eq:system_brunovsky}-\eqref{eq:xi_subsystem_transformed} in closed-loop with the controller \eqref{eq:proposed_control_mod_easy} is uniformly asymptotically stable. Uniform asymptotic stability of the origin of the system \eqref{eq:system_brunovsky}-\eqref{eq:xi_subsystem_transformed}-\eqref{eq:proposed_control_mod_easy} in the $(z,\xi)$-coordinates implies uniform asymptotic stability of the origin of the feedback equivalent system \eqref{eq:system}-\eqref{eq:aux_dynamics}-\eqref{eq:dynamic_control}-\eqref{eq:proposed_control_mod_easy} \cite{Sontag1998}.
\end{proof}
\begin{remark}
Similarly to Theorems \ref{th:main_result_easy} and \ref{th:main_result}, in Theorem \ref{th:main_result_astolfi}, we assume feedback linearizability through an embedding on some open set $U \subseteq \R^n$ without explicitly specifying under what conditions this occurs. This is done to give the user an opportunity to use either local or global results, depending on what their application allows for. To show local feedback linearizability through an embedding, we suggest using the conditions from Proposition 4 in \cite{Astolfi2018}, namely that there exist constants $w_j$, $j = 1, \hdots, n-1$ and an output map $h$ such that:
\begin{enumerate}[label=(\subscript{B}{{\arabic*}})]
    \item the matrix 
    \begin{align*}
        \mathcal{O}(x) =
        \begin{bmatrix}
        dh(x)^T & dL_fh(x)^T & \ldots & dL_f^{n-1}h(x)^T
        \end{bmatrix}^T,
    \end{align*}
    has rank $n$ at the origin, implying that $\Phi$ given by \eqref{eq:transformation} is a diffeomorphism from some neighborhood $U_1$ of the origin to $\Phi(U_1)$;
    \item $r(0) \neq 0$, which implies that $r(x) \neq 0$ for some neighborhood $U_2$ of the origin.
\end{enumerate}
These conditions imply that the system \eqref{eq:system} is feedback linearizable through an embedding on an open set $U = U_1 \cap U_2$.
Please note that the condition $(B_1)$ is also the sufficient condition for local observability at the origin. The condition $(B_2)$ is violated if and only if $L_g L_f^{j-1}h(0) = 0$ for all $j = 1, \ldots, n-1$, which is equivalent to \mbox{$\mathcal{O}(0) \cdot g (0) = 0$}. Given that $\mathcal{O}(0)$ is full-rank, this condition is, in turn, a consequence of $(B_1)$, provided $g(0) \neq 0$. 
\end{remark}
\begin{remark}
Note that the class of feedback linearizable systems through an embedding strictly contains feedback linearizable systems. In Section \ref{sec:ball_and_beam}, we provide an example of a system belonging to this class which is not feedback linearizable.
\end{remark}
\begin{remark}
In general, verifying condition $(A_2)$ can be a challenging task. Therefore, the authors of \cite{Astolfi2018} suggest substituting the ISS condition $(A_2)$ with the more verifiable condition that the matrix $A_\xi + A_w$ is Hurwitz, where:
\begin{align}
    A_w = \nabla_\xi 
    \begin{bmatrix}
        L_g h(x) \frac{s(x,\xi)}{r(x)}\bigg\rvert_{(x,\xi) = \Phi^{-1}(z,\xi)} \\
        L_g L_f h(x) \frac{s(x,\xi)}{r(x)} \bigg\rvert_{(x,\xi) = \Phi^{-1}(z,\xi)} \\
        \vdots \\
        L_g L^{n-2}_f h(x) \frac{s(x,\xi)}{r(x)} \bigg\rvert_{(x,\xi) = \Phi^{-1}(z,\xi)}
    \end{bmatrix}_{(z,\xi) = (0,0)}.
\end{align}
This is because the matrix $A_\xi + A_w$ is a linear approximation of the unforced $\xi$-subsystem in \eqref{eq:xi_subsystem_transformed} around the origin and its stability implies local ISS of the system \eqref{eq:xi_subsystem_transformed}.
\end{remark}

\section{Experiments and simulations}
\subsection{Quadrotor control experiment}
\label{sec:simulation}


\begin{figure}
    \centering
    \includegraphics[width=\linewidth]{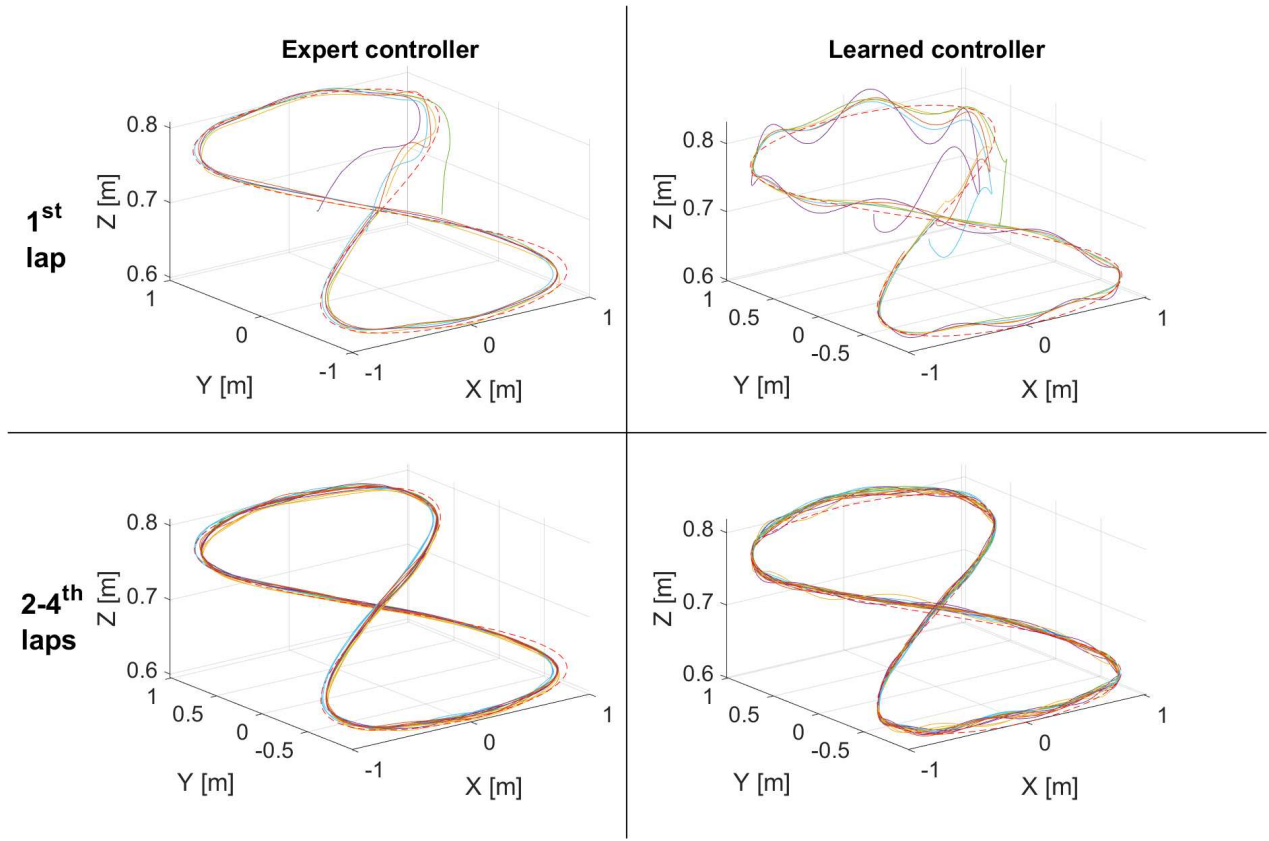}
    \caption{Trajectory tracking of the nonlinear controller from \cite{Faessler2015} (left column) and the learned controller from \eqref{eq:proposed_control_mod_easy} (right column) under five different initial conditions. Each experiment is plotted with a different color. The first lap trajectories (top row) are plotted separately from those in the subsequent laps (bottom row).}
    \label{fig:trajectory}
\end{figure}

\begin{figure*}
    \centering
    \vspace{4pt}
    \includegraphics[width=\linewidth]{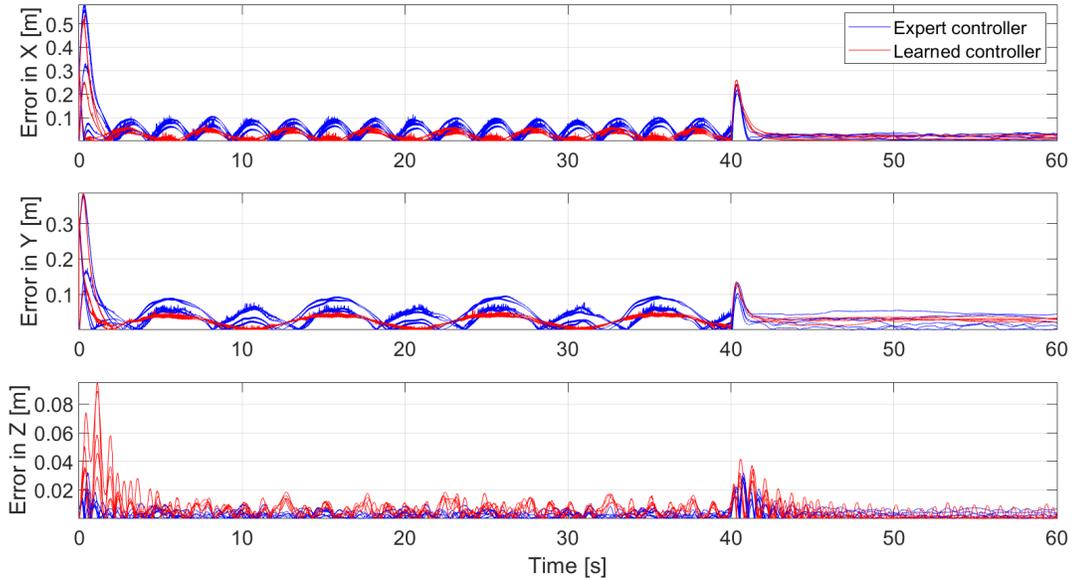}
    \caption{Comparison of tracking errors in $X$, $Y$ and $Z$ coordinates of learned controller from \eqref{eq:proposed_control_mod_easy} (red) and nonlinear controller from \cite{Faessler2015} (blue) for all five experiments.}
    \label{fig:tracking}
\end{figure*}

We illustrate the performance of our methodology using the example of quadrotor dynamics:
\begin{align}
    \label{eq:drone_1}
    \ddot{p} &= \frac{1}{m}\left( \tau R e_3  - [\omega]_{\times}J \omega \right),\\
    \label{eq:drone_2}
    \dot{R} &= R[\omega]_{\times},\\
    \label{eq:drone_3}
    \dot{\omega} &= J^{-1}(\eta - [\omega]_{\times}J\omega),
\end{align}
where: $p \in \R^3$, $R \in SO(3)$, $\omega \in \R^3$ are the position, orientation, and angular velocity of the quadrotor, respectively; $\tau \in \R$ and $\eta \in \R^3$ are thrust and torque inputs, respectively; $m \in \R$, and $J \in \R^{3 \times 3}$ are the mass and the inertia matrix; $[\; \cdot \;]_\times$ denotes the matrix form of the vector cross product, and $e_3 = \begin{bmatrix} 0 & 0 & 1 \end{bmatrix}^T$ is a unit vector. 

We split the dynamics \eqref{eq:drone_1}-\eqref{eq:drone_3} into two subsystems: one described by \eqref{eq:drone_1}-\eqref{eq:drone_2} with the state $x = (p,\dot{p},R,\tau)$ and the virtual inputs $u = (\dot{\tau}, \omega)$, and the other described by \eqref{eq:drone_3} with the state $x' = \omega$ and the virtual inputs $u' = \eta$. Typically, quadrotors have high-frequency internal controllers that track the desired angular velocity based on state feedback and, therefore, it is reasonable to assume that we can directly control the angular velocity \cite{Hehn2011}.

It is known that the dynamics \eqref{eq:drone_1}-\eqref{eq:drone_3} are differentially flat with respect to position and yaw angle \cite{Mellinger2011}. In what follows, we focus on controlling the position $p$, whereas the yaw angle is controlled to remain constant. Differential flatness allows us to transform the dynamics \eqref{eq:drone_1}-\eqref{eq:drone_2} into linear dynamics $\dot{z}_1 = z_2,\; \dot{z}_2 = z_3, \; \dot{z}_3 = v$
via a coordinate transformation:
\begin{align*}
    z = \begin{bmatrix}z_1 & z_2 & z_3 \end{bmatrix}^T \triangleq \begin{bmatrix}p & \dot{p} & \ddot{p} \end{bmatrix}^T,
\end{align*} 
and the feedback law:
\begin{align*}
    v = \frac{1}{m}(\dot{\tau}Re_3 - \tau R \omega_1 e_2 + \tau R \omega_2 e_1) \triangleq b(z)u. 
\end{align*}

We apply the controller design\footnote{The only difference of the expert controller used in this work and that used in \cite{Faessler2015} is that here the low-level controller is a linear PD controller.} from \cite{Faessler2015} to the dynamics \eqref{eq:drone_1}-\eqref{eq:drone_3} in simulation\footnote{We collected expert demonstrations in simulation to ensure there is no estimation error affecting the controller construction. In future work, we aim to construct the controller from expert solutions given by an observer.} and use the resulting solutions as the expert demonstrations. The controller parameters are chosen as follows: $K_P = \diag (7.0, 7.0, 16.5)$, $K_I = \diag (0, 0, 15.5)$, $K_D = \diag (5.0, 5.0, 3.4)$, $K_{rp} = 6.0$, $K_y = 2.0$. The expert is commanded to stabilize the quadrotor at the origin, starting from various positions, velocities and accelerations. Given the dimension of the state $z$ equals $9$, we record $10$ expert solutions $\{(z^i, v^i)\}_{i = 1}^{10}$ from simulations\footnote{The initial conditions used are the unit vectors of $\R^9$.}, including the pair corresponding to the trivial solution $(z^1, v^1) \equiv (0,0)$. Please note that the pairs $(z^i, v^i)$ in this context are merely evolutions of position, velocity, acceleration, and jerk. The recorded data is studied to ensure that the sufficient conditions of Theorem \ref{th:main_result_easy} are satisfied, i.e., the matrix $Z(t)$ in \eqref{eq:Z_matrix_easy} is always invertible and $\|Z(T)Z^{-1}(0)\| < 1$, and a fragment of length $T=2$ s is used to construct a stabilizing controller \eqref{eq:proposed_control_mod_easy}.

Next, we compare the learned controller \eqref{eq:proposed_control_mod_easy} and the expert controller from \cite{Faessler2015} by using them to control a BitCraze CrazyFlie 2.0 quadrotor. In these experiments, the control inputs $(\tau, \omega)$ are supplied by a computer via a USB radio at the average rate of $300$ Hz. The internal PD controller of the CrazyFlie tracks $(\tau, \omega)$ by controlling angular speeds of individual rotors. For state estimation, we use a Kalman filter that gets the position and attitude measurements from an OptiTrack motion capture system.

The experimental benchmark\footnote{Code used in the experiments can be found at \mbox{\url{https://github.com/cyphylab/cyphy_testbed/tree/LFD}}.} we choose to compare the controllers is to track the reference depicted on Figure \ref{fig:trajectory}, which consists of two parts: a figure of eight given by:
$$p_{R}(t) = \left(\sin{(4 \pi f t)}, \sin{(2 \pi f t)}, 0.1 \sin{(2 \pi f t)} + 0.7 \right),$$
where $f = 0.1$ Hz, from $t=0$ s to $t=40$ s; and a setpoint at the origin after $t \geq 40$ s. Note the reference trajectory is quite different from the collected expert demonstrations. We use the learned controller $\widehat{\kappa}$ from \eqref{eq:proposed_control_mod_easy} to control the tracking error with \mbox{$v(t) = \widehat{\kappa}(t, z(t)-z_R(t))$}, where $z_R = (p_R,\dot{p}_R, \ddot{p}_R)$, together with the feedback law:
\begin{equation*}
    \begin{aligned}
    u(t) & = (\dddot{p}_{R}(t) + v(t))/b(z(t)).
    \end{aligned}
\end{equation*}
For both the expert and the learned controllers, we perform five experiments --- each from a different initial position\footnote{The initial positions used are $(0, 0, 0.7)$, $(0.3, 0.3, 0.7)$, $(0.3, -0.3, 0.7)$, $(-0.3, 0.3, 0.7)$, $(-0.3, -0.3, 0.7)$.}. 

In Figure \ref{fig:trajectory}, we depict the quadrotor trajectories for both the nonlinear controller in \cite{Faessler2015} and the learned controller \eqref{eq:proposed_control_mod_easy} tracking the aforementioned trajectory. We plot the position trajectories in the first lap separately from those in the subsequent laps to decouple the transient behaviour of a controller from the steady-state behaviour. In Figure \ref{fig:tracking} we compare the tracking errors of the learned controller with those of the nonlinear controller from \cite{Faessler2015} for all five experiments. The initial conditions used during the experiments are purposefully chosen to be far from the initial conditions used during the simulation. The learned controller appears to track the trajectory well --- the error is of the order of centimeters. It can be seen qualitatively, however, from Figure \ref{fig:trajectory} that, in comparison to the expert controller, the learned controller takes a longer time to settle --- this is especially noticeable in the experiments where the initial position of the quadrotor does not match that of the reference. From Figure \ref{fig:tracking}, we observe that the errors of the learned controller and the expert controller are comparable, with the errors of the learned controller being slightly smaller in $X$ and $Y$ coordinates, whereas being slightly larger in $Z$ coordinates. For $t \geq 40$ s, the error in position does not tend to zero for neither of the controllers, which appears to contradict the theoretical results. We attribute this to the several milliseconds of delay with which the control input is sent to the quadrotor\footnote{Even in simulation, an introduction of such a delay into the control loop has resulted in the trajectory stabilizing at a non-zero steady-state error.}.


\subsection{Ball and beam control simulation}
\label{sec:ball_and_beam}
Consider the ball and beam model described by \cite{Hauser1992}:
\begin{equation}
\begin{aligned}
\label{eq:ball_beam}
    \ddot{r} &= \bar{b} (r \omega^2 - \bar{g} \sin(\phi)) \\
    \dot{\phi} &= \omega \\
    \dot{\omega} &= u,
\end{aligned}
\end{equation}
with $r \in \R$ and $v \in \R$ denoting the position and the velocity of the ball on the beam, respectively, while $\phi \in \R$ and $\omega \in \R$ denote the angle and the angular velocity of the beam with respect to the horizontal line, respectively. The constant $\bar{g}$ is the gravity constant, and the constant $\bar{b} = m/(J_b/R^2 + m)$, where $m$ is mass, $J_b$ is the moment of inertia, and $R$ is the radius of the ball. The state of the system \eqref{eq:ball_beam} is given by \mbox{$x = (r, \dot{r}, \phi, \omega)$}. The values of the parameters in this simulation are chosen to be $\bar{b} = 0.7143$ and $\bar{g} = 9.81$.

We choose the stabilizing controller\footnote{This controller is interesting to study because it is nonlinear, contains nested saturations, and utilizes backstepping.} proposed in \cite{Barbu1997} as the expert controller for the system \eqref{eq:ball_beam}. Since dimension of the state is $4$, we record simulations of $5$ expert solutions $\{(x^i, u^i)\}_{i=1}^5$ starting from various initial conditions\footnote{The initial conditions used are $(1,0,0,0)$, $(0,1,0,0)$, $(0,0,\pi/8, 0)$, $(0,0,0,10)$.}, including the trivial solution $(x^1, u^1) \equiv (0,0)$.

It can be shown that system \eqref{eq:ball_beam} is not feedback linearizable \cite{Hauser1992} and, therefore, techniques described in Sections \ref{sec:n+1_demons} and \ref{sec:more_than_n+1_demons} cannot be used. Instead, we use the technique described in Section \ref{sec:non_feedback} to approximate the expert controller. To immerse the system \eqref{eq:ball_beam}, we use the map $\Phi$ given by \eqref{eq:transformation} with:
\begin{align}
\label{eq:transformation_ball}
    \Phi_z(x,\xi) = 
    \begin{bmatrix}
    x_1 + \xi_1 \\
    x_2 + \xi_2 \\
    -\bar{b} \bar{g} \sin x_3 + \bar{b} x_1 x_4^2 + \xi_3 \\
    \bar{b} x_2 x_4^2 - \bar{g} x_4 \cos x_3 - \sum_{i = 1}^3 w_i \xi_i
    \end{bmatrix},
\end{align}
and the dynamic control law:
\begin{equation}
\label{eq:control_and_ext_ball}
\begin{aligned}
\begin{cases}
    u = \frac{1}{r(x)} \left( s(x,\xi) + v \right) \\
    \dot{\xi}_i = \xi_{i+1}, \quad i = 1,2 \\
    \dot{\xi}_3 = -w_1 \xi_1 - w_2 \xi_2 -w_3 \xi_3 - 2 x_1 x_4 u,
\end{cases}
\end{aligned}
\end{equation}
where:
\begin{align*}
    r(x) &= 2 \bar{b} x_2 x_4 - \bar{b} \bar{g} \cos x_3 + 2 w_3 b x_1 x_4 \\
    s(x, \xi) &= -\bar{b}^2 x_4^2 \left( -\bar{g} \sin x_3 + x_1 x_4^2 \right) - \bar{b} \bar{g} x_4^2 \sin x_3 + w_1 \xi_2\\
    &\quad + w_2 \xi_3 - w_1 w_3 \xi_1 - w_2 w_3 \xi_2 - w_3^2 \xi_3.
\end{align*}
Please note that this map and dynamic control law are only well-defined on an open set around the origin.

We choose the parameters $w = (1,3,3)$. First, we solve the differential equation in \eqref{eq:control_and_ext_ball} for the initial state $\xi(0) = 0$ using each of the previously collected expert solutions $\{(x^i, u^i) \}_{i=1}^5$ as inputs. Next, we use the transformations \eqref{eq:transform_z_astolfi} and \eqref{eq:transform_v_astolfi} to transform the expert solutions into the form $\{ (z^i, v^i) \}_{i=1}^5$. Then, these solutions are inspected to ensure that the conditions of Theorem \ref{th:main_result_astolfi} are satisfied, i.e., the matrix $Z(t)$ is always invertible and $\|Z(T) Z^{-1}(0) \| < 1$, and a fragment of length $T = 8 \text{s}$ is used to construct a stabilizing controller \eqref{eq:proposed_control_mod_easy}.

\begin{figure}
    \centering
    \includegraphics[width=\linewidth]{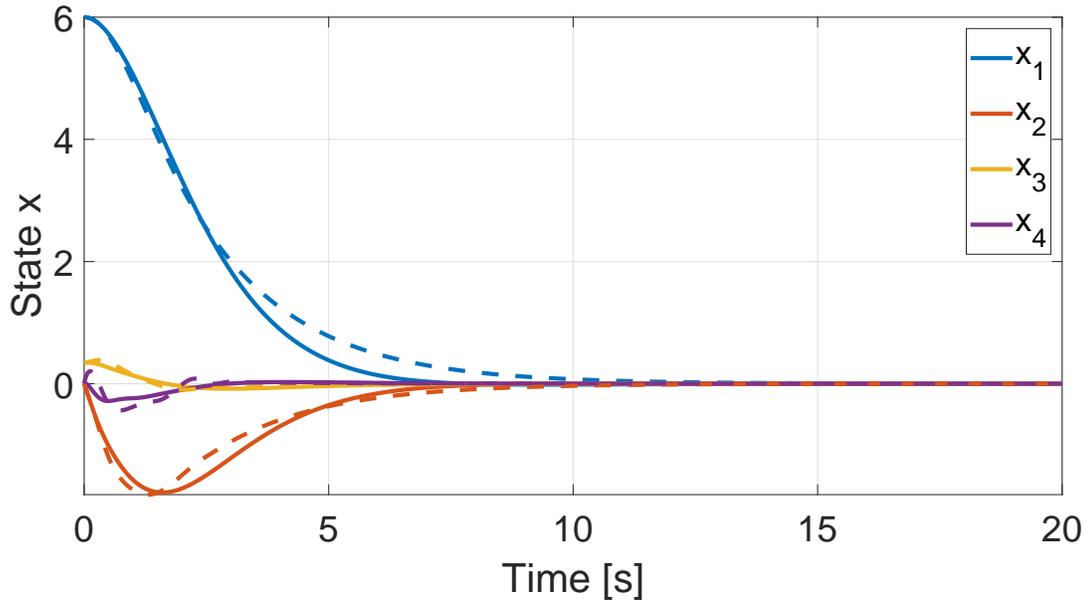}
    \caption{Comparison of stabilization between the learned controller from \eqref{eq:proposed_control_mod_easy} (solid lines) and nonlinear controller from \cite{Barbu1997} (dashed lines).}
    \label{fig:stabilization}
\end{figure}
\begin{figure}
    \centering
    \includegraphics[width=\linewidth]{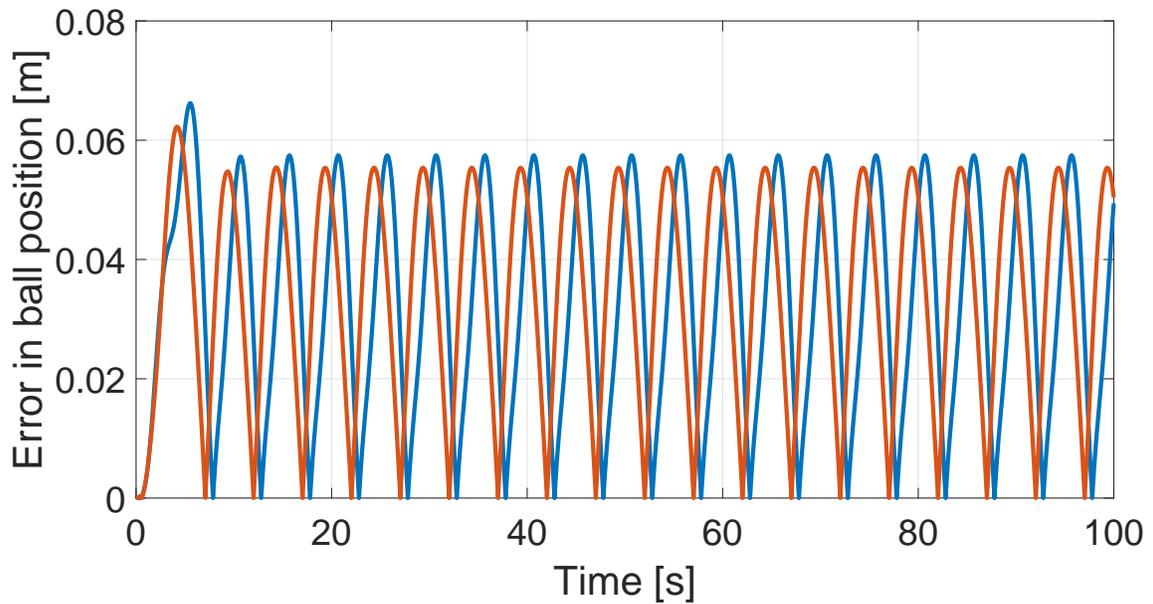}
    \caption{Comparison of tracking errors between the learned controller from \eqref{eq:proposed_control_mod_easy} (blue) and nonlinear controller from \cite{Barbu1997} (red).}
    \label{fig:trajectory_tracking_ball}
\end{figure}

We compare the learned controller \eqref{eq:proposed_control_mod_easy} and the expert controller from \cite{Barbu1997} based on their performance on two tasks: stabilization and output tracking. We define the position of the ball to be the output, i.e., $y = x_1$. When using the expert controller from \cite{Barbu1997} for output tracking, we use the output regulation method described in \cite{Isidori1990}. On Figure \ref{fig:stabilization}, we compare how both controllers stabilize the system to the origin from the initial state $x(0) = (6, 0, 0.345, 0)$. The initial conditions used when deploying the learned controller are purposefully chosen to be far from the initial conditions used during the expert's deployment. The closed-loop solutions of the learned controller and the expert controller appear to be very similar, although the learned controller stabilizes to the origin slightly slower than the expert. On Figure \ref{fig:trajectory_tracking_ball}, we compare how both controllers follow a reference trajectory given by $y_R(t) = 6 \cdot \cos \frac{2 \pi}{10} t$ from the initial state $x(0) = (6, 0, 0.345, 0)$ by plotting their tracking errors. We again observe that the performance of the learned controller and that of the expert controller are similar, with the learned controller having a slightly larger error.

\bibliography{references}
\bibliographystyle{IEEEtran}

\end{document}